\documentclass[conference,letterpaper]{IEEEtran}
\usepackage[cmex10]{amsmath}
\usepackage{amsthm}
\usepackage{mathtools}
\usepackage{amssymb}
\usepackage{dsfont}
\usepackage{xcolor}
\usepackage{float}
\usepackage{verbatim}
\usepackage{multirow}
\usepackage[maxbibnames=99,style=ieee,sorting=nyt,citestyle=numeric-comp]{biblatex}

\bibliography{bibliography.bib}

\title{On Algebraic Constructions of Neural Networks with Small Weights}
\author{%
   \IEEEauthorblockN{\textbf{Kordag Mehmet Kilic}, \textbf{Jin Sima} and \textbf{Jehoshua Bruck}}
   \IEEEauthorblockA{Electrical Engineering, California Institute of Technology, USA, \texttt{\{kkilic,jsima,bruck\}@caltech.edu}
   }
 }
\date{}

\newtheorem{theorem}{Theorem}
\newtheorem{corollary}{Corollary}[theorem]
\newtheorem{lemma}{Lemma}

\newtheorem{definition}{Definition}

\newcommand{\LT}{\mathcal{L}\mathcal{T}}

\usepackage{tikz}[border=2mm]
\usepackage{float}
\usetikzlibrary{shapes.geometric}
\AtBeginBibliography{\small}
\tikzset{
    ->,
    gate/.style={draw=black,fill=#1,minimum width=6mm,circle},
    square/.style={regular polygon,regular polygon sides=4},
    every pin edge/.style={draw=black}
}

\begin{document}

\maketitle

\begin{abstract}
    Neural gates compute functions based on weighted sums of the input variables. The expressive power of neural gates (number of distinct functions it can compute) depends on the weight sizes and, in general, large weights (exponential in the number of inputs) are required. Studying the trade-offs among the weight sizes, circuit size and depth is a well-studied topic both in circuit complexity theory and the practice of neural computation. We propose a new approach for studying these complexity trade-offs by considering a related algebraic framework. Specifically, given a single linear equation with arbitrary coefficients, we would like to express it using a system of linear equations with smaller (even constant) coefficients. The techniques we developed are based on Siegel’s Lemma for the bounds, anti-concentration inequalities for the existential results and extensions of Sylvester-type Hadamard matrices for the constructions.
    
    We explicitly construct a constant weight, optimal size matrix to compute the EQUALITY function (checking if two integers expressed in binary are equal). Computing EQUALITY with a single linear equation requires exponentially large weights. In addition, we prove the existence of the best-known weight size (linear) matrices to compute the COMPARISON function (comparing between two integers expressed in binary). In the context of the circuit complexity theory, our results improve the upper bounds on the weight sizes for the best-known circuit sizes for EQUALITY and COMPARISON.
\end{abstract}
\section{Introduction}
\label{sec:intro}

An $n$-input Boolean function is a mapping from $\{0,1\}^n$ to $\{0,1\}$. In other words, it is a partitioning of $n$-bit binary vectors into two sets with labels $0$ and $1$. In general, we can use systems of linear equations as descriptive models of these two sets of binary vectors. For example, the solution set of the equation $\sum_{i=1}^n x_i = k$ is the $n$-bit binary vectors $X = (x_1,\dots, x_n)$ where each $x_i \in \{0,1\}$ and $k$ is the number of $1$s in the vectors. We can ask three important questions: How expressive can a single linear equation be? How many equations do we need to describe a Boolean function?   Could we simulate a single equation by a system of equations with smaller integer weights?

Let us begin with an example: $3$-input PARITY function where we label binary vectors with odd number of 1s as 1. We can write it in the following form:
\begin{equation}
    \label{eq:parity_3}
    \text{PARITY}(X) = \mathds{1}\Big\{(2^2 x_3 + 2^1 x_2 + 2^0 x_1) \in \{1,2,4,7\} \Big\}
\end{equation}
where $\mathds{1}\{.\}$ is the indicator function with outputs 0 or 1. We express the binary vectors as integers by using binary expansions. Thus, it can be shown that if the weights are exponentially large in $n$, we can express all Boolean functions in this form.

Now, suppose that we are only allowed to use a single equality check in an indicator function. Considering the $3$-input PARITY function, we can simply obtain
\begin{equation}
    \label{eq:exp_construction}
    \begin{bmatrix}
        2^0 & 2^1 & 2^2 \\
        2^0 & 2^1 & 2^2 \\
        2^0 & 2^1 & 2^2 \\
        2^0 & 2^1 & 2^2
    \end{bmatrix}
    \begin{bmatrix}
        x_1 \\
        x_2 \\
        x_3
    \end{bmatrix}
    = \begin{bmatrix}
        1 \\
        2 \\
        4 \\
        7
    \end{bmatrix}
\end{equation}

None of the above equations can be satisfied if $X$ is labeled as 0. Conversely, if $X$ satisfies one of the above equations, we can label it as $1$. For an arbitrary Boolean function of $n$ inputs, if we list every integer associated with vectors labeled as $1$, the number of rows may become exponentially large in $n$. Nevertheless, in this fashion, we can compute this function by the following system of equations using smaller weights.
\begin{equation}
    \label{eq:parity_best}
    \begin{bmatrix}
        1 & 1 & 1 \\
        1 & 1 & 1
    \end{bmatrix}
    \begin{bmatrix}
        x_1 \\
        x_2 \\
        x_3
    \end{bmatrix}
    = \begin{bmatrix}
        1 \\
        3
    \end{bmatrix}
\end{equation}

Not only there is a simplification on the number of equations, but also the weights are reduced to smaller sizes. This phenomenon motivates the following question with more emphasis: For which Boolean functions could we obtain such simplifications in the number of equations and weight sizes? For PARITY, this simplification is possible because it is a \textit{symmetric} Boolean function, i.e., the output depends only on the number of $1$s of the input $X$. We are particularly interested in such simplifications from large weights to small weights for a class of Boolean functions called \textit{threshold functions}. Note that we use the word ``large'' to refer to exponentially large quantities in $n$ and the word ``small'' to refer to polynomially large quantities (including $O(n^0) = O(1)$) in $n$.

\subsection{Threshold Functions and Neural Networks}
Threshold functions are commonly treated functions in Boolean analysis and machine learning as they form the basis of neural networks. Threshold functions compute a weighted summation of binary inputs and feed it to a \textit{threshold} or \textit{equality} check. If this sum is fed to the former, we call the functions \textit{linear threshold functions} (see \eqref{eq:lt}). If it is fed to the latter, we call them \textit{exact threshold functions} (see \eqref{eq:elt}). We can write an $n$-input threshold function using the indicator function where $w_i$s are integer weights and $b$ is a \textit{bias} term.
\begin{align}
    \label{eq:lt}
    f_{\LT} (X) &= \mathds{1}\Big\{\sum_{i=1}^n w_i x_i \geq b\Big\} \\
    \label{eq:elt}
    f_{\mathcal{E}} (X) &= \mathds{1}\Big\{\sum_{i=1}^n w_i x_i = b\Big\}
\end{align}

A device computing the corresponding threshold function is called a \textit{gate} of that type. To illustrate the concept, we define COMPARISON (denoted by COMP) function which computes whether an $n$-bit integer $X$ is greater than another $n$-bit integer Y. The exact counterpart of it is defined as the EQUALITY (denoted by EQ) function which checks if two $n$-bit integers are equal (see Figure \ref{fig:comp_and_eq}). For example, we can write the EQ function in the following form.
\begin{equation}
\label{eq:equality}
\text{EQ}(X,Y) = \mathds{1}\Big\{\sum_{i=1}^n 2^{i-1}(x_i-y_i) = 0 \Big\}
\end{equation}

\begin{figure}[H]
    \begin{tikzpicture}
    \tikzstyle{sum} = [gate=white,label=center:+]
    
    \tikzstyle{input} = [circle]
    
    \newcommand{\nodenum}{3}
    \newcommand{\eq}{=} 
    \pgfmathsetmacro{\offset}{\nodenum}
    \node[gate=white,label=center:$\LT$,pin=right:COMP] (lt) at (2,-\nodenum-0.5) {};

    \foreach \x in {1,...,\nodenum}
    {   
        \pgfmathsetmacro{\exponent}{int(\x-1)}
        \node[input,label=180:$x_\x$] (xo-\x) at (0,-\x) {};
        \draw (xo-\x) -- node[above,pos=0.3] {$2^{\exponent}$} (lt);
    }
    \foreach \y in {1,...,\nodenum}
    {   
        \pgfmathsetmacro{\exponent}{int(\y-1)}
        \node[input,label=180:$y_\y$] (yo-\y) at (0,-\y - \offset) {};
        \draw (yo-\y) -- node[above,pos=0.25] {$-2^{\exponent}$} (lt);
    }
    
    \pgfmathsetmacro{\offsetx}{4.6}
    \node[gate=white,label=center:$\mathcal{E}$,pin=right:EQ] (e) at (2+\offsetx,-\nodenum-0.5) {};
    \foreach \x in {1,...,\nodenum}
    {   
        \pgfmathsetmacro{\exponent}{int(\x-1)}
        \node[input,label=180:$x_\x$] (xo-\x) at (\offsetx,-\x) {};
        \draw (xo-\x) -- node[above,pos=0.3] {$2^{\exponent}$} (e);
    }
    \foreach \y in {1,...,\nodenum}
    {   
        \pgfmathsetmacro{\exponent}{int(\y-1)}
        \node[input,label=180:$y_\y$] (yo-\y) at (\offsetx,-\y - \offset) {};
        \draw (yo-\y) -- node[above,pos=0.25] {$-2^{\exponent}$} (e);
    }
\end{tikzpicture}
\caption{The $3$-input COMP and EQ functions for integers $X$ and $Y$ computed by linear threshold and exact threshold gates. A gate with an $\LT$(or $\mathcal{E}$) inside is a linear (or exact) threshold gate. More explicitly, we can write $\text{COMP}(X,Y) = \mathds{1}\{4x_3 + 2x_2 + x_1 \geq 4y_3 + 2y_2 + y_1\}$ and $\text{EQ}(X,Y) = \mathds{1}\{4x_3 + 2x_2 + x_1 = 4y_3 + 2y_2 + y_1\}$}
\label{fig:comp_and_eq}
\end{figure}
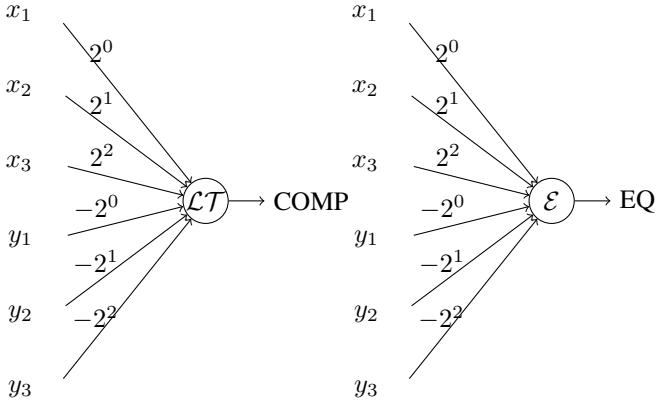

In general, it is proven that the weights of a threshold function can be represented by $O(n\log{n})$-bits and this is tight \cite{alon1997anti,babai2010weights, haastad1994size,muroga1971threshold}. However, it is possible to construct ``small'' weight \textit{threshold circuits} to compute any threshold function \cite{amano2005complexity,goldmann1993simulating,hofmeister1996note, siu1991power}. This transformation from a circuit of depth $d$ with exponentially large weights in $n$ to another circuit with polynomially large weights in $n$ is typically within a constant factor of depth (e.g. $d+1$ or $3d + 3$ depending on the context) \cite{goldmann1993simulating,vardi2020neural}. For instance, such a transformation would simply follow if we can replace any ''large`` weight threshold function with ``small'' weight depth-$2$ circuits so that the new depth becomes $2d$.

It is possible to reduce polynomial size weights into constant weights by replicating the gates that is fed to the top gate recursively (see Figure \ref{fig:replicate}). Nevertheless, this would inevitably introduce a polynomial size blow-up in the circuit size. We emphasize that our focus is to achieve this weight size reduction from polynomial weights to constant weights with at most a constant size blow-up in the circuit size.

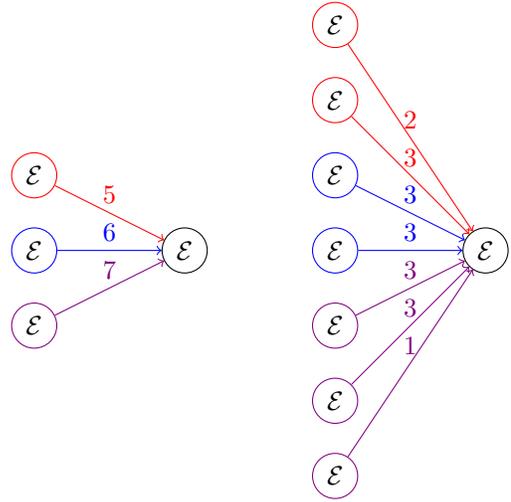
\begin{figure}
    \centering
    \begin{tikzpicture}
        \tikzstyle{exa} = [gate=white,label=center:$\mathcal{E}$]
        \node[exa,draw=red] (i-1) at (0,1) {};
        \node[exa,draw=blue] (i-2) at (0,0) {};
        \node[exa,draw=violet] (i-3) at (0,-1) {};
        \node[exa] (o-1) at (2, 0) {};
        \draw[draw=red] (i-1) -- node[above,color=red] {$5$} (o-1);
        \draw[draw=blue] (i-2) -- node[above,color=blue] {$6$}(o-1);
        \draw[draw=violet] (i-3) -- node[above,color=violet] {$7$}(o-1);
        
        \pgfmathsetmacro{\offset}{4}
        
        \node[exa,draw=red] (i-11) at (\offset,3) {};
        \node[exa,draw=red] (i-12) at (\offset,2) {};
        \node[exa,draw=blue] (i-21) at (\offset,1) {};
        \node[exa,draw=blue] (i-22) at (\offset,0) {};
        \node[exa,draw=violet] (i-31) at (\offset,-1) {};
        \node[exa,draw=violet] (i-32) at (\offset,-2) {};
        \node[exa,draw=violet] (i-33) at (\offset,-3) {};
        
        \node[exa] (o-11) at (\offset + 2, 0) {};
        
        \draw[draw=red] (i-11) -- node[above,color=red] {$2$} (o-11);
        \draw[draw=red] (i-12) -- node[above,color=red] {$3$} (o-11);
        \draw[draw=blue] (i-21) -- node[above,color=blue] {$3$}(o-11);
        \draw[draw=blue] (i-22) -- node[above,color=blue] {$3$}(o-11);
        \draw[draw=violet] (i-31) -- node[above,color=violet] {$3$}(o-11);
        \draw[draw=violet] (i-32) -- node[above,color=violet] {$3$}(o-11);
        \draw[draw=violet] (i-33) -- node[above,color=violet] {$1$}(o-11);
    \end{tikzpicture}
    \caption{An example of a weight transformation for a single gate (in black) to construct constant weight circuits. Different gates are colored in red, blue, and violet and depending on the weight size, each gate is replicated a number of times. In this example, each weight in this construction is at most 3.}
    \label{fig:replicate}
\end{figure}

For neural networks and learning tasks, the weights are typically finite precision real numbers. To make the computations faster and more power efficient, the weights can be quantized to small integers with a loss in the accuracy of the network output \cite{jacob2018quantization, hubara2016binarized}. In practice, given that the neural circuit is large in size and deep in depth, this loss in the accuracy might be tolerated. We are interested in the amount of trade-off in the increase of size and depth while using as small weights as possible. More specifically, our goal is to provide insight to the computation of single threshold functions with large weights using threshold circuits with small weights by relating theoretical upper and lower bounds to the best known constructions.

In this manner, for the ternary weight case ($\{-1,0,1\}$ as in \cite{li2016ternary}), we give an explicit and optimal size circuit construction for the EQ function using depth-2 circuits. This optimality is guaranteed by achieving the theoretical lower bounds asymptotically up to vanishing terms \cite{kilic2021neural,roychowdhury1994lower}. We also prove an existential result on the COMP constructions to reduce the weight size on the best known constructions. It is not known if constant weight constructions exist without a polynomial blow-up in the circuit size and an increase in the depth for arbitrary threshold functions.

\subsection{Bijective Mappings from Finite Fields to Integers}
It seems that choosing powers of two as the weight set is important. In fact, we can expand any arbitrary weight in binary and use powers of two as a fixed weight set for any threshold function \cite{kilic2021neural}. This weight set is a choice of convenience and a simple solution for the following question: How small could the elements of $\mathcal{W} = \{w_1, w_2, \cdots, w_n\}$ be if $\mathcal{W}$ has all distinct subset sums (DSS)? If the weights satisfy this property, called the \textbf{DSS property}, they can define a bijection between $\{0,1\}^n$ and integers. Erd\H{o}s conjectured that the largest weight $w \in \mathcal{W}$ is upper bounded by $c_0 2^n$ for some $c_0 > 0$ and therefore, choosing powers of two as the weight set is asymptotically optimal. The best known result for such weight sets yields $0.22002 \cdot 2^{n}$ and currently, the best lower bound is $\Omega(2^n/\sqrt{n})$ \cite{bohman1998construction,dubroff2021note,guy2004unsolved}. 

Now, let us consider the following linear equation where the weights are fixed to the ascending powers of two but $x_i$s are not necessarily binary. We denote the powers of two by the vector $w_b$.
\begin{equation}
    \label{eq:w_b}
    w_b^T x = \sum_{i=1}^n 2^{i-1} x_i = 0
\end{equation}
As the weights of $w_b$ define a bijection between $n$-bit binary vectors and integers, $w_b^T x = 0$ does not admit a non-trivial solution for the alphabet $\{-1,0,1\}^n$. This is a necessary and sufficient condition to compute the EQ function given in \eqref{eq:equality}. We extend this property to $m$ many rows to define \textit{\textup{EQ} matrices} which give a bijection between $\{0,1\}^n$ and $\mathbb{Z}^m$. Thus, an EQ matrix can be used to compute the EQ function in \eqref{eq:equality} and our goal is to use smaller weight sizes in the matrix.

\begin{definition}
A matrix $A \in \mathbb{Z}^{m\times n}$ is an $\textup{EQ matrix}$ if the homogeneous system $Ax = 0$ has no non-trivial solutions in $\{-1,0,1\}^n$.
\end{definition}

Let $A \in \mathbb{Z}^{m\times n}$ be an EQ matrix with the weight constraint $W \in \mathbb{Z}$ such that $|a_{ij}| \leq W$ for all $i,j$ and let $R$ denote the \textit{rate} of the matrix $A$, which is $n/m$. It is clear that any full-rank square matrix is an EQ matrix with $R = 1$. Given any $W$, how large can this $R$ be? For the maximal rate, a necessary condition can be proven by Siegel's Lemma \cite{siegel2014einige}.

\begin{lemma}[Siegel's Lemma (modified)]
Consider any integer matrix $A \in \mathbb{Z}^{m \times n}$ with $m < n$ and $|a_{ij}| \leq W$ for all $i,j$ and some integer $W$. Then, $Ax = 0$ has a non-trivial solution for an integer vector $x \in \mathbb{Z}^n$ such that $||x||_\infty \leq (\sqrt{n}W)^\frac{m}{n-m}$.
\end{lemma}

It is shown that if $m = o(n/\log{nW})$, then any $A \in \mathbb{Z}^{m \times n}$ with weight constraint $W$ admits a non-trivial solution in $\{-1,0,1\}^n$ and cannot be an EQ matrix, i.e., $R = O(\log{nW})$ is tight \cite{kilic2021neural}. A similar result can be obtained by the matrix generalizations of Erd\H{o}s' Distinct Subset Sum problem \cite{costa2021variations}.

When $m = O(n/\log{nW})$, the story is different. If $m = O(n/\log{n})$, there exists a matrix $A \in \{-1,1\}^{m \times n}$ such that every non-trivial solution of $Ax = 0$ satisfies $\max_{j} |x_j| \geq c_0 \sqrt{n}^\frac{m}{n-m}$ for a positive constant $c_0$. This is given by Beck's Converse Theorem on Siegel's Lemma \cite{beck2017siegel}.

For an explicit construction, it is possible to achieve the optimal rate $R = O(\log{nW})$ if we allow $W = poly(n)$. This can be done by the Chinese Remainder Theorem (CRT) and the Prime Number Theorem (PNT) \cite{amano2005complexity,hofmeister1996note,kilic2021neural}. It is known that CRT can be used to define \textit{residue codes} \cite{mandelbaum1972error}. For an integer $x$ and modulo base $p$, we denote the modulo operation by $[x]_p$, which maps the integer to a value in $\{0,...,p-1\}$. Suppose $0 \leq Z < 2^n$ for an integer $Z$. One can encode this integer by the $m$-tuple $(d_1,d_2,\cdots,d_m)$ where $[Z]_{p_i} = d_i$ for a prime number $p_i$. Since we can also encode $Z$ by its binary expansion, the CRT gives a bijection between $\mathbb{Z}^m$ and $\{0,1\}^n$ as long as $p_1 \cdots p_m > 2^n$. By taking modulo $p_i$ of Equation \eqref{eq:w_b}, we can obtain the following matrix, defined as a \textit{CRT matrix}:

\begin{align}
\label{eq:ex_crt}
&\begin{bmatrix*}[l]
        [2^0]_{3} & [2^1]_{3} & [2^2]_{3} & [2^4]_{3} & [2^5]_{3} & [2^6]_3 & [2^7]_3 \\
        [2^0]_{5} & [2^1]_{5} & [2^2]_{5} & [2^4]_{5} & [2^5]_{5} & [2^6]_5 & [2^7]_5 \\
        [2^0]_{7} & [2^1]_{7} & [2^2]_{7} & [2^4]_{7} & [2^5]_{7} & [2^6]_7 & [2^7]_7 \\
        [2^0]_{11} & [2^1]_{11} & [2^2]_{11} & [2^4]_{11} & [2^5]_{11} & [2^6]_{11} & [2^7]_{11}
        \end{bmatrix*} \\
        &\hphantom{aaaaaaaaaa}= 
        \begin{bmatrix*}[r]
        1 & 2 & 1 & 2 & 1 & 2 & 1 & 2 \\
        1 & 2 & 4 & 3 & 1 & 2 & 4 & 3 \\
        1 & 2 & 4 & 1 & 2 & 4 & 1 & 2 \\
        1 & 2 & 4 & 8 & 5 & 10 & 9 & 7
        \end{bmatrix*}_{4 \times 8}
\end{align}

We have $Z < 256$ since $n=8$ and $3\cdot 5\cdot 7\cdot 11 = 1155$. Therefore, this CRT matrix is an EQ matrix. In general, by the PNT, one needs $O(n/\log{n})$ many rows to ensure that $p_1 \cdots p_m > 2^n$. Moreover, $W$ is bounded by the maximum prime size $p_m$, which is $O(n)$ again by the PNT. 

However, it is known that constant weight EQ matrices with asymptotically optimal rate exist by Beck's Converse Theorem on Siegel's Lemma \cite{beck2017siegel,kilic2021neural}. In this paper, we give an explicit construction where $W = 1$ and asymptotic efficiency in rate is achieved up to vanishing terms. It is in fact an extension of Sylvester-type Hadamard matrices.

\begin{equation}
    \label{eq:eq_4x8}
    \begin{bmatrix*}[r]
        1 &  1 &  1 &  1 &  1 &  1 & 1 & 0 \\
        1 & -1 &  0 &  1 & -1 &  0 & 0 & 1 \\ 
        1 &  1 &  1 & -1 & -1 & -1 & 0 & 0 \\
        1 & -1 &  0 & -1 &  1 &  0 & 0 & 0
    \end{bmatrix*}_{4\times 8}
\end{equation}

One can verify that the matrix in \eqref{eq:eq_4x8} is an EQ matrix and its rate is twice the trivial rate being the same in \eqref{eq:ex_crt}. Therefore, due to the optimality in the weight size, this construction can replace the CRT matrices in the constructions of EQ matrices.

We can also focus on $q$-ary representation of integers in a similar fashion by extending the definition of EQ matrices to this setting. In our analysis, we always treat $q$ as a constant value.

\begin{definition}
A matrix $A \in \mathbb{Z}^{m\times n}$ is an $\textup{EQ}_q\textup{ matrix}$ if the homogeneous system $Ax = 0$ has no non-trivial solutions in $\{-q+1,\dots,q-1\}^n$.
\end{definition}

If $q = 2$, then we drop $q$ from the notation and say that the matrix is an EQ matrix. For the $\text{EQ}_q$ matrices, the optimal rate given by Siegel's Lemma is still $R = O(\log{nW})$ and constant weight constructions exist. We give an extension of our construction to $\text{EQ}_q$ matrices where $W = 1$ and asymptotic efficiency in rate is achieved up to constant terms.

\subsection{Maximum Distance Separable Extensions of EQ Matrices}
Residue codes are treated as Maximum Distance Separable (MDS) codes because one can extend the CRT matrix by adding more prime numbers to the matrix (without increasing $n$) so that the resulting integer code achieves the Singleton bound in Hamming distance \cite{tay2015non}. However, we do not say a CRT matrix is an \textit{MDS matrix} as this refers to another concept.

\begin{definition}
An integer matrix $A \in \mathbb{Z}^{m \times n}$ ($m \leq n$) is \textup{MDS} if and only if no $m \times m$ submatrix $B$ is singular.
\end{definition}

\begin{definition}
\label{def:rmds}
An integer matrix $A \in \mathbb{Z}^{rm \times n}$ is \textup{MDS for $q$-ary bijections with MDS rate $r$ and EQ rate $R = n/m$} if and only if every $m \times n$ submatrix $B$ is an $\textup{EQ}_q$ matrix.
\end{definition}

Because Definition \ref{def:rmds} considers solutions over a \textbf{restricted} alphabet, we denote such matrices as $\textit{RMDS}_q$. Remarkably, as $q \rightarrow \infty$, both MDS definitions become the same. Similar to the $\text{EQ}_q$ definition, we drop the $q$ from the notation when $q = 2$. A CRT matrix is not MDS, however, it can be $\text{RMDS}_q$.

We can demonstrate the difference of both MDS definitions by the following matrix. This matrix is an RMDS matrix with EQ rate $2$ and MDS rate $5/4$ because any $4 \times 8$ submatrix is an EQ matrix. This is in fact the same matrix in \eqref{eq:ex_crt} with an additional row with entries $[2^{i-1}]_{13}$ for $i \in \{1,...,8\}$.

\begin{align}
    \begin{bmatrix*}[r]
        1 & 2 & 1 & 2 & 1 & 2 & 1 & 2 \\
        1 & 2 & 4 & 3 & 1 & 2 & 4 & 3 \\
        1 & 2 & 4 & 1 & 2 & 4 & 1 & 2 \\
        1 & 2 & 4 & 8 & 5 & 10 & 9 & 7 \\
        1 & 2 & 4 & 8 & 3 & 6 & 12 & 11
        \end{bmatrix*}_{5 \times 8}
\end{align}

Here, the determinant of the $5\times5$ submatrix given by the first five columns is 0. Thus, this matrix is not an MDS matrix.

In this work, we are interested in $\text{RMDS}_q$ matrices. We show that for a constant weight $\text{RMDS}_q$ matrix, the MDS rate bound $r = O(1)$ is tight. We also provide an existence result for such $\text{RMDS}_q$ matrices where the weight size is bounded by $O(r)$ given that the EQ rate is $O(\log{n})$.

The following is a summary of our contributions in this paper.

\begin{itemize}
    \item In Section \ref{sec:constr}, we explicitly give a rate-efficient $\text{EQ}$ matrix construction with constant entries $\{-1,0,1\}$ where the optimality is guaranteed up to vanishing terms. This solves an open problem in \cite{kilic2021neural} and \cite{roychowdhury1994lower}.
    \item In Section \ref{sec:rmds}, we prove that the MDS rate $r$ of an $\text{RMDS}_q$ matrix with entries from an alphabet $\mathcal{Q}$ with cardinality $k$ should satisfy $r \leq k^{k+1}$. Therefore, constant weight $\text{RMDS}_q$ matrices can at most achieve $r = O(1)$.
    \item In Section \ref{sec:rmds}, provide an existence result for $\text{RMDS}_q$ matrices given that $W = O(r)$ and the optimal EQ rate $O(\log{n})$. In contrast, the best known results give $W = O(rn)$ with the optimal EQ rate $O(\log{n})$.
    \item In Section \ref{sec:neural}, we apply our results to Circuit Complexity Theory to obtain better weight sizes with asymptotically \textbf{no trade-off} in the circuit size for the depth-2 EQ and COMP constructions, as shown in the Table \ref{tab:cct}.
\end{itemize}

\begin{table}[h!]
\caption{Results for the Depth-2 Circuit Constructions}
\label{tab:cct}
\centering
\begin{tabular}{|c||c|c||c|c|}
     \hline
     \multirow{2}{*}{\textbf{Function}} & \multicolumn{2}{|c||}{\textbf{This Work}} & \multicolumn{2}{|c|}{\textbf{Previous Works}} \\ 
     \cline{2-3}\cline{4-5} & Weight Size & Constructive & Weight Size & Constructive \\
     \hline\hline
     \multirow{2}{*}{EQ} & \multirow{2}{*}{$O(1)$} & \multirow{2}{*}{Yes} & $O(1)$\cite{kilic2021neural} & No \\
     \cline{4-5} & & & $O(n)$\cite{roychowdhury1994lower} & Yes \\
     \hline
     COMP & $O(n)$ & No & $O(n^2)$\cite{amano2005complexity} & Yes \\
     \hline
\end{tabular}
\end{table}

\section{Rate-efficient Constructions with Constant Alphabet Size}
\label{sec:constr}

The $m\times n$ EQ matrix construction we give here is based on Sylvester's construction of Hadamard matrices. It is an extension of it to achieve higher rates with the trade-off that the matrix is no longer full-rank.

\begin{theorem}
\label{th:constr}
Suppose we are given an EQ matrix $A_0 \in \{-1,0,1\}^{m_0\times n_0}$. At iteration $k$, we construct the following matrix $A_k$:
\begin{equation}
    A_k = \begin{bmatrix*}[c]
        A_{k-1} &  A_{k-1} & I_{m_{k-1}} \\
        A_{k-1} & -A_{k-1} &  0 
    \end{bmatrix*}
\end{equation}
$A_k$ is an EQ matrix with $m_k = 2^k m_0$, $n_k = 2^k n_0 (\frac{k}{2}\frac{m_0}{n_0} + 1)$ for any integer $k \geq 0$.
\end{theorem}
\begin{proof}
We will apply induction. The case $k = 0$ is trivial by assumption. For the system $A_kx = z$, let us partition the vector $x$ and $z$ in the following way:
\begin{equation}
    \label{eq:constr_partition}
    \begin{bmatrix*}[c]
        A_{k-1} &  A_{k-1} & I_{m_{k-1}} \\
        A_{k-1} & -A_{k-1} &  0 
    \end{bmatrix*}
    \begin{bmatrix*}
        x^{(1)} \\
        x^{(2)} \\
        x^{(3)}
    \end{bmatrix*}
    = \begin{bmatrix*}
        z^{(1)} \\
        z^{(2)}
    \end{bmatrix*}
\end{equation}
Then, setting $z = 0$, we have $A_{k-1}x^{(1)} = A_{k-1}x^{(2)}$ by the second row block. Hence, the first row block of the construction implies $2A_{k-1}x^{(1)} + x^{(3)} = 0$. Each entry of $x^{(3)}$ is either $0$ or a multiple of $2$. Since $x_i \in \{-1,0,1\}$, we see that $x^{(3)} = 0$. Then, we obtain $A_{k-1}x^{(1)} = A_{k-1}x^{(2)} = 0$. Applying the induction hypothesis on $A_{k-1}x^{(1)}$ and $A_{k-1}x^{(2)}$, we see that $A_kx = 0$ admits a unique trivial solution in $\{-1,0,1\}^{n_k}$.

To see that $m_k = 2^k m_0$ is not difficult. For $n_k$, we can apply induction.
\end{proof}

To construct the matrix given in \eqref{eq:eq_4x8}, we can use $A_0 = I_1 = \begin{bmatrix}1\end{bmatrix}$, which is a trivial rate EQ matrix, and take $k = 2$. By rearrangement of the columns, we also see that there is a $4\times 4$ Sylvester-type Hadamard matrix in this matrix.


For the sake of completeness, we will revisit Siegel's Lemma to prove the best possible rate one can obtain for an EQ matrix with weights $\{-1,0,1\}$ as similarly done in \cite{kilic2021neural}. We note that the following lemma gives the best possible rate upper bound to the best of our knowledge and our construction asymptotically shows the sharpness of Siegel's Lemma similar to Beck's result.

\begin{lemma}
\label{lem:upper}
For any $\{-1,0,1\}^{m\times n}$ EQ matrix, $R = \frac{n}{m} \leq \frac{1}{2}\log{n} + 1$.
\end{lemma}
\begin{proof}
By Siegel's Lemma, we know that for a matrix $A \in \mathbb{Z}^{m\times n}$, the homogeneous system $Ax = 0$ attains a non-trivial solution in $\{-1,0,1\}^n$ when
\begin{equation}
    ||x||_\infty \leq (\sqrt{n}W)^{\frac{m}{n-m}} \leq 2^{1-\epsilon}
\end{equation}
for some $\epsilon > 0$. Then, since $W = 1$, we deduce that $
    R = \frac{n}{m} \geq \frac{1}{2(1-\epsilon)}\log{n} + 1
$. Obviously, an EQ matrix cannot obtain such a rate. Taking $\epsilon \rightarrow 0$, we conclude the best upper bound is $
    R \leq \frac{1}{2}\log{n} + 1
$.
\end{proof}

Using Lemma \ref{lem:upper} for our construction with $A_0 = \begin{bmatrix}1\end{bmatrix}$, we compute the upper bound on the rate and the real rate
\begin{align}
    R_{upper} &= \frac{k + 1 + \log{(k+2)}}{2},\hphantom{a} R_{constr} = \frac{k}{2}+1 \\
    \frac{R_{upper}}{R_{constr}} &= 1 + \frac{\log{(k+2)}-1}{k+2} \implies \frac{R_{upper}}{R_{constr}} \sim 1
\end{align}
which concludes that the construction is \textbf{optimal} in rate up to vanishing terms. By deleting columns, we can generalize the result to any $n \in \mathbb{Z}$ to achieve optimality up to a constant factor of $2$.

We conjecture that one can extend the columns of any Hadamard matrix to obtain an EQ matrix with entries $\{-1,0,1\}$ achieving rate optimality up to vanishing terms. In this case, the Hadamard conjecture implies a rich set of EQ matrix constructions.

Given $Ax = z \in \mathbb{Z}^m$, we note that there is a linear time decoding algorithm to find $x \in \{0,1\}^n$ uniquely, given in the appendix. This construction can also be generalized to $\text{EQ}_q$ matrices with the same proof idea. In this case, the rate is optimal up to a factor of $q$.

\begin{theorem}
\label{th:constr_q}
Suppose we are given an $\text{EQ}_q$ matrix $A_0 \in \{-1,0,1\}^{m_0\times n_0}$. At iteration $k$, we construct the following matrix $A_k$:
\begin{equation}
    \begin{bmatrix*}[c]
        A_{k-1} &  A_{k-1} &  A_{k-1} & \cdots  & A_{k-1} & I_{m_{k-1}} \\
        A_{k-1} & -A_{k-1} &  0 &\cdots & 0 & 0 \\
        0 &  A_{k-1} & -A_{k-1} & \cdots & 0 & 0 \\
        \vdots & \vdots & \vdots & \ddots &  \vdots & \vdots \\
        0 &  0 &  0 & \cdots & -A_{k-1} & 0
    \end{bmatrix*}
\end{equation}
$A_k$ is an $\text{EQ}_q$ matrix with $m_k = q^k m_0$, $n_k = q^k n_0 (\frac{k}{q}\frac{m_0}{n_0} + 1)$ for any integer $k \geq 0$.
\end{theorem}

\section{Bounds on the Rate and the Weight Size of $\text{RMDS}_q$ Matrices}
\label{sec:rmds}

Similar to a CRT-based RMDS matrix, is there a way to extend the EQ matrix given in Section \ref{sec:constr} without a large trade-off in the weight size? We give an upper bound on the MDS rate $r$ based on the alphabet size of an $\text{RMDS}_q$ matrix.

\begin{theorem}
\label{th:mds_lower_bound}
An $\text{RMDS}_q$ matrix $A \in \mathbb{Z}^{rm \times n}$ with entries from $\mathcal{Q} = \{q_1,...,q_k \}$ satisfies $r \leq k^{k+1}$ given that $n > k$.
\end{theorem}
\begin{proof}
The proof is based on a simple counting argument. We rearrange the rows of a matrix in blocks $B_i$ for $i \in \{1,\cdots,k^{k+1}\}$ assuming $n > k$. Here, each block contains rows starting with a prefix of length $k+1$ in the lexicographical order of the indices, i.e.,
\begin{equation}
    A = \begin{bmatrix}
        B_1 \\
        B_2 \\
        \vdots \\
        B_{k^{k+1}}
    \end{bmatrix}
    = 
    \begin{bmatrix}
        q_1 & q_1 & \cdots & q_{1} & \cdots \\
        q_1 & q_1 & \cdots & q_{2} & \cdots \\
        \vdots   &  \vdots  & \ddots & \vdots & \ddots \\
        q_{k} & q_{k} & \cdots & q_{k} & \cdots
    \end{bmatrix}
\end{equation}

For instance, the block $B_1$ contains the rows starting with the prefix $\begin{bmatrix}
    q_1 & q_1 & \cdots & q_{1}
\end{bmatrix}_{1 \times k+1}$. It is easy to see that there is a vector $x \in \{-1,0,1\}^n$ that will make at least one of the $B_i x = 0$ because it is guaranteed that the $(k + 1)$-th column will be equal to the one of the preceding elements by Pigeonhole Principle.

Since the matrix is MDS, any $m$ row selections should be an EQ matrix. Therefore, any $B_i$ should not contain more than $m$ rows. Again, by Pigeonhole Principle, $rm \leq k^{k+1}m$ and consequently, $r \leq k^{k+1}$.
\end{proof}

The condition that $n > k$ is trivial in the sense that if the weights are constant, then $k$ is a constant. Therefore, the MDS rate can only be a constant due to Theorem \ref{th:mds_lower_bound}.

\begin{corollary}
An $\text{RMDS}_q$ matrix $A \in \mathbb{Z}^{rm \times n}$ weight size $W = O(1)$ can at most achieve the MDS rate $r = O(1)$.
\end{corollary} 

In Section \ref{sec:constr}, we saw that CRT-based EQ matrix constructions are not optimal in terms of the weight size. For $\text{RMDS}_q$ matrices, we now show that the weight size can be reduced by a factor of $n$. The idea is to use the probabilistic method on the existence of $\text{RMDS}_q$ matrices.

\begin{lemma}
\label{lem:prob}
Suppose that $a \in \{-W,\dots,W\}^n$ is uniformly distributed and $x_i \in \{-q+1,\dots,q-1\} \setminus 0$ for $i \in \{1,\dots,n\}$ are fixed where $q$ is a constant. Then, for some constant $C$,
\begin{equation}
    Pr(a^T x = 0) \leq \frac{C}{\sqrt{n}W}
\end{equation}
\end{lemma}
\begin{proof}
We start with a statement of the Berry-Esseen Theorem.

\begin{lemma}[Berry-Esseen Theorem]
Let $X_1,\dots,X_n$ be independent centered random variables with finite third moments $\mathbb{E}[|X_i^3|] = \rho_i$ and let $\sigma^2 = \sum_{i=1}^n \mathbb{E}[X_i^2]$. Then, for any $t > 0$,
\begin{equation}
    \label{eq:berry_esseen}
    \Big| Pr\Big(\sum_{i=1}^n X_i \leq t \Big) - \Phi(t) \Big| \leq C \sigma^{-3} \sum_{i=1}^n \rho_i
\end{equation}
where $C$ is an absolute constant and $\Phi(t)$ is the cumulative distribution function of $\mathcal{N}(0,\sigma^2)$.
\end{lemma}

Since the density of a normal random variable is uniformly bounded by $1/\sqrt{2\pi\sigma^2}$, we obtain the following.
\begin{equation}
    Pr\Big(\Big|\sum_{i=1}^n X_i \Big| \leq t \Big) \leq \frac{2t}{\sqrt{2\pi \sigma^2}} + 2C \sigma^{-3} \sum_{i=1}^n \rho_i
\end{equation}

Let $a^{(n-1)}$ denote the $(a_1,\dots,a_{n-1})$ and similarly, let $x^{(n-1)}$ denote $(x_1,\dots,x_{n-1})$. By the Total Probability Theorem, we have
\begin{align}
    Pr(&a^T x = 0) = Pr\Bigg(\frac{{a^{(n-1)}}^T x^{(n-1)}}{|x_n|} \in \{-W,\dots,W\}\Bigg) \nonumber \\ 
    &Pr\Bigg(a^T x = 0 \Big| \frac{{a^{(n-1)}}^T x^{(n-1)}}{|x_n|} \in \{-W,\dots,W\}\Bigg) \nonumber \\
    & = Pr\Bigg(\frac{{a^{(n-1)}}^T x^{(n-1)}}{|x_n|} \in \{-W,\dots,W\}\Bigg)\frac{1}{2W+1}
\end{align}
where the last line follows from the fact that $Pr(a_n = k)$ for some $k \in \{-W,\dots,W\}$ is $\frac{1}{2W+1}$. The other conditional probability term is $0$ because $Pr(a_n = k)$ for any $k \not\in \{-W,\dots,W\}$ is $0$. We will apply Berry-Esseen Theorem to find an upper bound on the first term.

We note that $\mathbb{E}[a_i^2x_i^2] = x_i^2 \frac{(2W+1)^2-1}{12}$ and $\mathbb{E}[|a_i^3x_i^3|] = |x_i|^3 \frac{2}{2W+1}\Big(\frac{W(W+1)}{2}\Big)^2$. Then,

\begin{align}
    Pr\Bigg(&\Bigg|\frac{{a^{(n-1)}}^T x^{(n-1)}}{|x_n|}\Bigg| \leq W \Bigg) \leq \frac{2W|x_n|}{\sqrt{2\pi}\sqrt{\frac{(2W+1)^2-1}{12}\sum_{i=1}^{n-1} x_i^2}} \nonumber \\ 
    &\hphantom{0000000000} + C \frac{\frac{2}{2W+1}\Big(\frac{W(W+1)}{2}\Big)^2 \sum_{i=1}^{n-1} |x_i|^3}{\Big(\frac{(2W+1)^2-1}{12}\Big)^\frac{3}{2}\Big(\sum_{i=1}^{n-1} x_i^2\Big)^\frac{3}{2}}
\end{align}

One can check that the whole RHS is in the form of $\frac{C'}{\sqrt{n-1}}$ for a constant $C'$ as we assume that $q$ is a constant. We can bound $|x_i|$s in the numerator by $q-1$ and $|x_i|$s in the denominator by $1$. The order of $W$ terms are the same and we can use a limit argument to bound them as well. Therefore, for some $C'' > 0$,

\begin{equation}
    Pr(a^T x = 0) \leq \frac{C'}{\sqrt{n-1}}\frac{1}{2W+1} \leq \frac{C''}{\sqrt{n}W}
\end{equation}
\end{proof}

Lemma \ref{lem:prob} is related to the Littlewood-Offord problem and anti-concentration inequalities are typically used in this framework \cite{halasz1977estimates, rudelson2008littlewood}. We specifically use Berry-Esseen Theorem. Building on this lemma and the union bound, we have the following theorem.

\begin{theorem}
\label{th:rmds}
An $\text{RMDS}_q$ matrix $M \in \mathbb{Z}^{rm \times n}$ with entries in  $\{-W,\dots,W\}$ exists if $m = \Omega(n/\log{n})$ and $W = O(r)$.
\end{theorem}
\begin{proof}

Let $A$ be any $m \times n$ submatrix of $M$. Then, by the union bound (as done in \cite{karingula2021singularity}), we have
\begin{align}
    Pr(\text{M is not RMDS}_q) = P &\leq \binom{rm}{m} Pr(\text{A is not EQ}_q) \nonumber \\
    & \leq r^m e^m Pr(\text{A is not EQ}_q)
\end{align}
We again use the union bound to sum over all $x \in \{-q+1,\dots,q-1\}^n \setminus 0$ the event that $Ax = 0$ by counting non-zero entries in $x$ by $k$. Also, notice the independence of the events that $a_i^T x = 0$ for $i \in \{1,\dots,m\}$. Let $x^{(k)}$ denote an arbitrary $x \in \{-q+1,\dots,q-1\}^n \setminus 0$ with $k$ non-zero entries. Then,

\begin{align}
    Pr(\text{A is not EQ}_q) &\leq \sum_{x \in \{-q+1,\dots,q-1\}^n \setminus 0} \prod_{i=1}^m Pr(a_i^Tx = 0) \\
    &\leq \sum_{k=1}^n \binom{n}{k} (2(q-1))^k Pr(a^Tx^{(k)} = 0)^m
\end{align}

We can use Lemma \ref{lem:prob} to bound the $Pr(a^Tx^{(k)} = 0)$ term. Therefore,

\begin{align}
    P &\leq r^m e^m \sum_{k=1}^n \binom{n}{k} (2(q-1))^k \Big(\frac{C}{\sqrt{k}W}\Big)^m \\
    \label{eq:binomial_terms}
    &\hphantom{000}=\sum_{k=1}^{T-1} \binom{n}{k} (2(q-1))^k \Big(\frac{rC_1}{\sqrt{k}W}\Big)^m \nonumber \\
    &\hphantom{000000} + \sum_{k=T}^{n} \binom{n}{k} (2(q-1))^k \Big(\frac{rC_1}{\sqrt{k}W}\Big)^m
\end{align}
for some $C_1, T \in \mathbb{Z}$. We bound the first summation by using $\binom{n}{k}(2(q-1))^k \leq (n(q-1))^k$ and choosing $k = 1$ for the probability term. This gives a geometric sum from $k = 1$ to $k = T-1$. 

\begin{align}
    \sum_{k=1}^{T-1} (n(q-1))^k \Big(\frac{rC_1}{W}\Big)^m = \Big(\frac{rC_1}{W}\Big)^m \Big(\frac{(n(q-1))^T - 1}{n(q-1) - 1} \Big)
\end{align}

For the second term, we take the highest value term $\frac{rC_1}{\sqrt{T}W}$ and $\sum_{k=T}^n \binom{n}{k} 2(q-1)^k \leq \sum_{k=0}^n \binom{n}{k} 2(q-1)^k = (2(q-1) + 1)^n$. Hence,

\begin{align}
    P & \leq \Big(\frac{rC_1}{W}\Big)^m C_2 (n(q-1))^T + (2(q-1)+1)^n \Big(\frac{rC_1}{\sqrt{T}W}\Big)^m \\
    & \leq 2^{C_3 T\log{n(q-1)} - m\log{W/rC_1}} \nonumber \\
    &\hphantom{00000}+ 2^{n\log{(2(q-1)+1)} - m \log{(\sqrt{T}W/rC_1)}}
\end{align}

We take $T = O(n^c)$ for some $0 < c < 1$. It is easy to see that if $m = \Omega(n/\log{n})$ and $W = O(r)$, both terms vanish as $n$ goes to the infinity.

\end{proof}
When $r = 1$, we prove an existential result on EQ matrices already known in \cite{kilic2021neural}. We remark that the proof technique for Theorem \ref{th:rmds} is not powerful enough to obtain non-trivial bounds on $W$ when $m = 1$ to attack Erd\H{o}s' conjecture on the DSS weight sets.

The CRT gives an explicit way to construct $\text{RMDS}_q \in \mathbb{Z}^{rm\times n}$ matrices. We obtain $p_{rm} = O(rm\log{rm}) = O(rn)$ given that $r = O(n^c)$ for some $c > 0$ and $m = O(n/\log{n})$ by the PNT. Therefore, we have a factor of $O(n)$ weight size reduction in Theorem \ref{th:rmds}. However, modular arithmetical properties of the CRT do not reflect to $\text{RMDS}_q$ matrices. Therefore, an $\text{RMDS}_q$ matrix cannot replace a CRT matrix in general (see Appendix).

\section{Applications of the Algebraic Results to Neural Circuits}
\label{sec:neural}

In this section, we will give EQ and COMP threshold circuit constructions. We note that in our analysis, the size of the bias term is ignored.

One can construct a depth-$2$ exact threshold circuit with small weights to compute the EQ function \cite{kilic2021neural}. For the first layer, we select the weights for each exact threshold gate as the rows of the EQ matrix. Then, we connect the outputs of the first layer to the top gate, which just computes the $m$-input AND function (i.e. $\mathds{1}\{z_1 + ... + z_m = m\}$ for $z_i \in \{0,1\}$). In Figure \ref{fig:eq_constr}, we give an example of EQ constructions.

\begin{figure}[h]
    \centering
    \begin{tikzpicture}
    \tikzstyle{sum} = [gate=white,label=center:+]
    \tikzstyle{exa} = [gate=white,label=center:$\mathcal{E}$]
    \node[exa,pin=right:EQ] (out) at (4,-2) {};
    \tikzstyle{input} = [circle]
    \newcommand{\inputnum}{8}
    \newcommand{\domnum}{4}
    \newcommand{\eq}{=}
    \newcommand{\offset}{1}
    \foreach \x in {1,...,\inputnum}
    {   
        \node[input,label=180:$x_\x'$] (i-\x) at (0,-\x*.5) {};
        
    }
    
    \foreach \x in {1,...,\domnum}
    {   
        \node[exa] (e-\x) at (2.5,{-\x*(\offset)+0.25}) {};
        
        \draw (e-\x) -- (out);
    }
    
    \node[input] (bias) at (3.5,-1) {};
    \draw (bias) -- node[right,pos=0.3] {$-4$} (out);
    
    \def\eqmatrixorig{{{1, 1, 1, 1,1, 1,1,0},
                       {1,-1, 1,-1,0, 0,0,1},
                       {1, 1,-1,-1,1,-1,0,0},
                       {1,-1,-1, 1,0, 0,0,0}
                     }}
    \def\eqmatrix{{{1, 1, 1, 1,1, 1,0,1},
                   {1, 0, 1, 0,0,-1,1,-1},
                   {1, 0,-1,-1,1,-1,0,1},
                   {1, 0,-1, 0,0, 1,0,-1}
                   }}     
    \foreach \x in {1,...,\inputnum}
    {
        \foreach \y in {1,...,\domnum}
        {
            \pgfmathtruncatemacro{\val}{\eqmatrix[\y-1][\x-1]}
            \ifnum\val=0{}
            \else{
                \ifnum\val=1{
                    \draw (i-\x) -- (e-\y);
                 }
                \else{
                  \draw[color=red] (i-\x) -- (e-\y);
                 }\fi
            }
            \fi
        }
    }
\end{tikzpicture}
    \caption{An example of $\text{EQ}(X,Y)$ constructions with 8 $x_i' = x_i - y_i$ inputs (or 16 if $x_i$s and $y_i$s are counted separately) and 5 exact threshold gates (including the top gate). The black(or red) edges correspond to the edges with weight 1 (or -1).}
    \label{fig:eq_constr}
\end{figure}
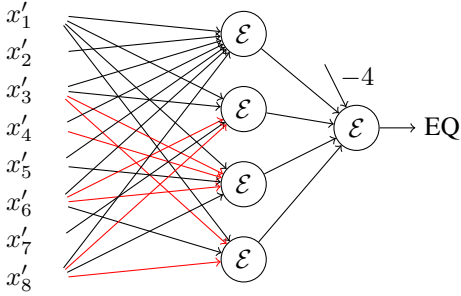

We roughly follow the previous works for the construction of the COMP \cite{amano2005complexity}. First, let us define $F^{(l)}(X,Y) = \sum_{i=l+1}^n 2^{i-l-1} (x_i-y_i)$ so that $\mathds{1}\{F^{(l)}(X,Y) \geq 0\}$ is a $(n-l)$-bit COMP function for $l < n$ and $F^{(0)}(X,Y) = F(X,Y)$. We say that $X \geq Y$ when $F(X,Y) \geq 0$ and vice versa. We also denote by $X^{(l)}$ the most significant $(n-l)$-tuple of an $n$-tuple vector $X$. We have the following observation.

\begin{lemma}
\label{lem:comp}
Let $F^{(l)}(X,Y) = \sum_{i=l+1}^n 2^{i-l-1}(x_i - y_i)$ and $F^{(0)}(X,Y) = F(X,Y)$ for $X, Y \in \{0,1\}^n$. Then,
\begin{align}
    F(X,Y) > 0  &\Leftrightarrow \exists l : F^{(l)}(X,Y) = \hphantom{-}1 \\
    F(X,Y) < 0 &\Leftrightarrow \exists l : F^{(l)}(X,Y) = -1
\end{align}
\end{lemma}
\begin{proof}
It is easy to see that if $X-Y = (0,0,\dots,0,1,\times,\dots,\times)$ where the vector has a number of leading 0s and $\times$s denote any of $\{-1,0,1\}$, we see that $F(X,Y) > 0$ (this is called the \textit{domination property}). Similarly, for $F(X,Y) < 0$, the vector $X-Y$ should have the form $(0,0,\dots,0,-1,\times,\dots,\times)$ and $F(X,Y) = 0$ if and only if $X-Y = (0,\dots,0)$. The converse holds similarly.
\end{proof}

By searching for the $(n-l)$-tuple vectors $(X-Y)^{(l)} = (0,...,0,-1)$ for all $l \in \{0,...,n-1\}$, we can compute the COMP. We claim that if we have an $\text{RMDS}_q$ matrix $A \in \mathbb{Z}^{rm \times n}$, we can detect such vectors by solving $A^{(l)} (X-Y)^{(l)} = -a_{n-l}$ where $A^{(l)}$ is a truncated version of $A$ with the first $n-l$ columns and $a_{n-l}$ is the $(n-l)$-th column. Specifically, we obtain the following:

\begin{align}
    \label{eq:comp_forward}
    X < Y &\Rightarrow \sum_{l=0}^{n-1} \mathds{1}\{A^{(l)}(X-Y)^{(l)} = -a_{n-l}\} \geq rm \\
    \label{eq:comp_backward}
    X \geq Y &\Rightarrow \sum_{l=0}^{n-1} \mathds{1}\{A^{(l)}(X-Y)^{(l)} = -a_{n-l}\} < n(m-1)
\end{align}

Here, the indicator function works row-wise, i.e., we have the output vector $z \in \{0,1\}^{rmn}$ such that $z_{rml + i} = \mathds{1}\{(A^{(l)}(X-Y)^{(l)})_i = -(a_{n-l})_i\}$ for $i \in \{1,\dots,rm\}$ and $l \in \{0,\dots,n-1\}$. We use an $\text{RMDS}_3$ matrix in the construction to map $\{-1,0,1\}^{n-l}$ vectors bijectively to integer vectors with large Hamming distance. Note that each exact threshold function can be replaced by two linear threshold functions and a summation layer (i.e. $\mathds{1}\{F(X) = 0\} = \mathds{1}\{F(X) \geq 0\} + \mathds{1}\{-F(X) \geq 0\} - 1)$ which can be absorbed to the top gate (see Figure \ref{fig:absorption}) . This increases the circuit size only by a factor of two.

\begin{figure}
    \centering
    \begin{tikzpicture}
        \tikzstyle{input} = [circle]
        \tikzstyle{redg} = [gate=white,label=center:$\mathcal{L}\mathcal{T}$,draw=red]
        \tikzstyle{blueg} = [gate=white,label=center:$\mathcal{L}\mathcal{T}$,draw=blue]
        \tikzstyle{sum} = [gate=white,label=center:+]
        \pgfmathsetmacro{\nodenum}{3}
        \pgfmathsetmacro{\offset}{2.75}
        \newcommand{\eq}{=}
    \node[redg,label=45:\textcolor{red}{$\mathds{1}\{F(X) \geq 0\}$}] (ex) at (1.25,-3) {};
    \node at (-.5,-4.9) {$\vdots$};
    \foreach \x in {1,...,\nodenum}
    {   \pgfmathsetmacro{\inputlabels}{\x<\nodenum ? int(\x-1) : "n"}
        \pgfmathsetmacro{\xp}{\x<\nodenum ? \x : "L"}
        \ifnum\x=\nodenum {\node[input,label=180:$x_n$] (xo-\x) at (0,{-\x*0.75-div(\x,\nodenum)*0.75-\offset}) {};}\else{\node[input,label=180:$x_\inputlabels$] (xo-\x) at (0,{-\x*0.75-div(\x,\nodenum)*0.75-\offset}) {};}\fi
        \draw (xo-\x) -- node[above,pos=0.3] {} (ex);
    }
        \pgfmathsetmacro{\offset}{4}
        \node[blueg,label=315:\textcolor{blue}{$\mathds{1}\{F(X) \leq 0\}$}] (ex_neg) at (1.25,-2-\offset) {};
    \foreach \x in {1,...,\nodenum}
    {   \pgfmathsetmacro{\inputlabels}{\x<\nodenum ? int(\x-1) : "n"}
        \pgfmathsetmacro{\xp}{\x<\nodenum ? \x : "L"}
        \draw (xo-\x) -- node[above,pos=0.3] {} (ex_neg);
    }
    \node[sum,pin=0:$\mathds{1}\{F(X) \eq 0\}$] (out) at (3,-4.5) {};
    \node[input] (bias) at (2.75, -3.5) {};
    
    \draw (ex) -- node[above,pos=0.6] {$1$} (out);
    \draw (ex_neg) -- node[above,pos=0.6] {$1$} (out);
    \draw (bias) -- node[above,pos=0.3] {$-1$} (out);
    \end{tikzpicture}
    \caption{A construction of an arbitrary exact threshold function ($\mathds{1}\{F(X) = 0\}$) using two linear threshold functions ($\mathds{1}\{F(X) \geq 0\}$ and $\mathds{1}\{F(X) \leq 0\}$) and a summation node. This summation node can be removed if its output is connected to another gate due to linearity.}
    \label{fig:absorption}
\end{figure}
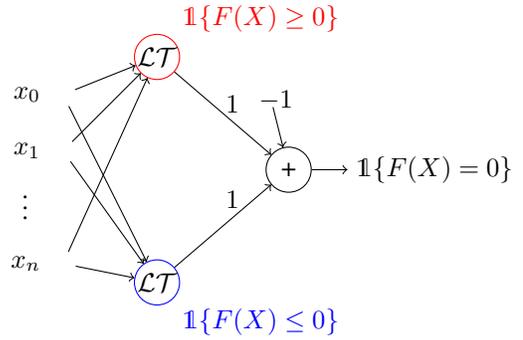

If $F(X,Y) < 0$, $z_{rml +i}$ should be 1 for all $i \in \{1,\dots,rm\}$ for some $l$ by Lemma \ref{lem:comp}. For $F(X,Y) \geq 0$ and any $l$, the maximum number of 1s that can appear in $z_{rml + i}$ is upper bounded by $m-1$ because A is $\text{RMDS}_3$. Therefore, the maximum number of 1s that can appear in $z$ is upper bounded by $n(m-1)$. A sketch of the construction is given in Figure \ref{fig:comp_constr}.
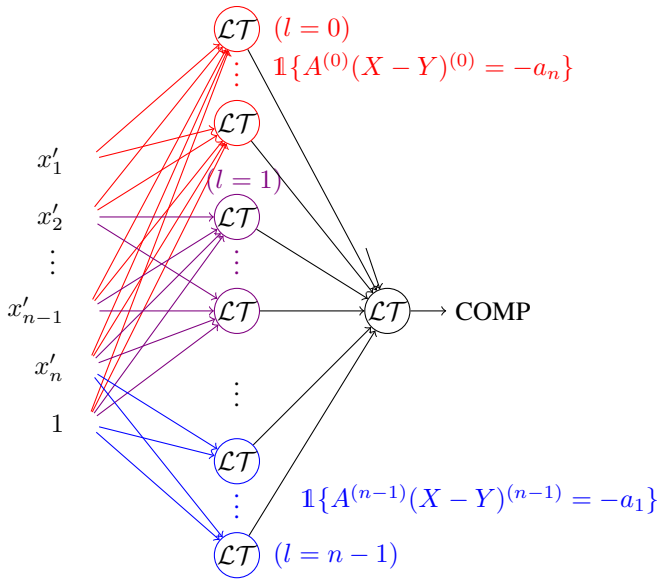
\begin{figure}[h]
    \centering
    \begin{tikzpicture}
    \tikzstyle{sum} = [gate=white,label=center:+]
    \tikzstyle{lt} = [gate=white,label=center:$\mathcal{L}\mathcal{T}$]
    \node[lt,pin=right:COMP] (out) at (4,-2) {};
    \tikzstyle{input} = [circle]
    \newcommand{\inputnum}{5}
    \newcommand{\Lnum}{3}
    \newcommand{\domnum}{2}
    \newcommand{\eq}{=}
    \newcommand{\offset}{1}
    
    \def\colors{red,violet,blue}
    
    \node[input,label=180:$1$] (i-5) at (0,-3.5) {};
    \node[input,label=180:$x_1'$] (i-1) at (0,0) {};
    \node[input,label=180:$x_2'$] (i-2) at (0,-0.75) {};
    \node at (-0.45,-1.25) {$\vdots$};
    \node[input,label=180:$x_{n-1}'$] (i-3) at (0,-2) {};
    \node[input,label=180:$x_n'$] (i-4) at (0,-2.75) {};
    
    \node[input,label=0:\textcolor{red}{$\hphantom{a}(l\eq 0)$}] (l-0) at (2,1.75) {};
    \node[lt,draw=red] (lt-1-1) at (2,1.75) {};
    \node[color=red,label=right:\textcolor{red}{$\hphantom{a}\mathds{1}\{A^{(0)}(X-Y)^{(0)} = -a_n\}$}] at (2,1.25) {$\vdots$};
    \node[lt,draw=red] (lt-1-2) at (2,0.5) {};
    
    \node[input,label=\textcolor{violet}{$\hphantom{a}(l\eq 1)$}] (l-0) at (2,-0.75) {};
    \node[lt,draw=violet] (lt-2-1) at (2,-0.75) {};
    \node[color=violet] at (2,-1.25) {$\vdots$};
    \node[lt,draw=violet] (lt-2-2) at (2,-2) {};
    
    \node at (2, -3) {$\vdots$};
    
    \node[input,label=0:\textcolor{blue}{$\hphantom{a}(l\eq n-1)$}] (l-0) at (2,-5.25) {};
    \node[lt,draw=blue] (lt-3-1) at (2,-4) {};
    \node[color=blue,label=right:\textcolor{blue}{$\hphantom{aaa}\mathds{1}\{A^{(n-1)}(X-Y)^{(n-1)} = -a_1\}$}] at (2,-4.5) {$\vdots$};
    \node[lt,draw=blue] (lt-3-2) at (2,-5.25) {};
    
    \foreach[count=\l] \c in \colors 
    {
        \foreach \y in {1,...,\domnum}
        {
            \foreach \x in {1,...,\inputnum}
            {
                \pgfmathsetmacro{\i}{int(\x + \l - 1 + div(\l,\Lnum)}
                \ifnum\i<6
                {
                    \draw[color=\c] (i-\i) -- (lt-\l-\y) {};
                }
                \else{}
                \fi
            }
            \draw (lt-\l-\y) -- node[midway] {} (out);
        }
    }
    
    \node[input] (bias) at (3.66,-1) {};
    \draw (bias) -- node[right,pos=0.3] {} (out);
\end{tikzpicture}
    \caption{A sketch of the $\text{COMP}(X,Y)$ construction where $x_i' = x_i - y_i$ using linear threshold gates. Each color specifies an $l$ value in the construction. If $X < Y$, all the $rm$ many gates in at least one of the colors will give all $1$s at the output. Otherwise, all the $rm$ many gates in a color will give at most $(m-1)$ many $1$s at the output.}
    \label{fig:comp_constr}
\end{figure}

In order to make both cases separable, we choose $r = n$. At the second layer, the top gate is a MAJORITY gate ($\mathds{1}\{\sum_{i=1}^{rmn} z_i < n(m-1)\} = \mathds{1}\{\sum_{i=1}^{rmn} -z_i \geq -n(m-1)+1\}$). By Theorem \ref{th:rmds}, there exists an $\text{RMDS}_3$ matrix with $m = O(n/\log{n})$ and $W = O(n)$. Thus, there are $rmn + 1 = O(n^3/\log{n})$ many gates in the circuit, which is the best known result. The same size complexity can be achieved by a CRT-based approach with $W = O(n^2)$ \cite{amano2005complexity}.

\section{Conclusion}

We explicitly constructed a rate-efficient constant weight EQ matrix for a previously existential result. Namely, we proved that the CRT matrix is not optimal in weight size to compute the EQ function and obtained the optimal EQ function constructions. For the COMP function, the weight size complexity is improved by a linear factor, using $\text{RMDS}_q$ matrices and their existence. An open problem is whether similar algebraic constructions can be found for general threshold functions so that these ideas can be developed into a weight quantization technique for neural networks.

\section*{Acknowledgement}
This research was partially supported by The Carver Mead New Adventure Fund.

\printbibliography

\appendix

\textbf{The Decoding Algorithm:} Suppose that $A_kx = z$ is given where $A_k$ is an $m\times n$ EQ matrix given in Theorem \ref{th:constr} starting with $A_0 = \begin{bmatrix} 1 \end{bmatrix}$. One can obtain a linear time decoding algorithm in $n$ by exploiting the recursive structure of the construction. Let us partition $x$ and $z$ in the way given in \eqref{eq:constr_partition}. It is clear that $x^{(3)} = [z^{(1)} + z^{(2)}]_2$. Also, after computing $x^{(3)}$, we find that $A_{k-1}x^{(1)} = (z^{(1)} + z^{(2)} - x^{(3)})/2$ and $A_{k-1}x^{(2)} = (z^{(1)} - z^{(2)} - x^{(3)})/2$. These operations can be done in $O(m_{k-1})$ time complexity. Let $T(m)$ denote the time to decode $z \in \mathbb{Z}^m$. Then, $T(m) = 2T(m/2) + O(m)$ and by the Master Theorem, $T(m) = O(m\log{m}) = O(n)$.

\textbf{An Example for the Decoding:} Consider the following system Ax = b:

\begin{equation*}
    \begin{bmatrix*}[r]
        1 &  1 &  1 &  1 &  1 &  1 & 1 & 0 \\
        1 & -1 &  0 &  1 & -1 &  0 & 0 & 1 \\ 
        1 &  1 &  1 & -1 & -1 & -1 & 0 & 0 \\
        1 & -1 &  0 & -1 &  1 &  0 & 0 & 0
    \end{bmatrix*}
    \begin{bmatrix}
        x_1 \\
        x_2 \\
        x_3 \\
        x_4 \\
        x_5 \\
        x_6 \\
        x_7 \\
        x_8
    \end{bmatrix} = \begin{bmatrix*}
        4 \\
        -2 \\
        -1 \\
        0
    \end{bmatrix*}
\end{equation*}

Here is the first step of the algorithm:
\begin{equation*}
    a = \begin{bmatrix} 4 \\ -2\end{bmatrix},\hphantom{a} b = \begin{bmatrix*} -1 \\ 0\end{bmatrix*} \Rightarrow \begin{bmatrix} x_7 \\ x_8 \end{bmatrix} = [a + b]_2 = \begin{bmatrix*}[r] [3]_2 \\ [-2]_2 \end{bmatrix*} = \begin{bmatrix} 1 \\ 0\end{bmatrix}
\end{equation*}
Then, we obtain the two following systems:
\begin{align*}
    &\begin{bmatrix*}[r]
        1 &  1 &  1 \\
        1 & -1 &  0
    \end{bmatrix*}
    \begin{bmatrix}
        x_1 \\
        x_2 \\
        x_3 \\
    \end{bmatrix} = \frac{a + b - \begin{bmatrix} 1 \\ 0 \end{bmatrix}}{2} = \begin{bmatrix*}
        1 \\
        -1
    \end{bmatrix*} \\
    &\begin{bmatrix*}[r]
        1 &  1 &  1 \\
        1 & -1 &  0
    \end{bmatrix*}
    \begin{bmatrix}
        x_4 \\
        x_5 \\
        x_6 \\
    \end{bmatrix} = \frac{a - b - \begin{bmatrix} 1 \\ 0 \end{bmatrix}}{2} = \begin{bmatrix*}
        2 \\
        -1
    \end{bmatrix*}
\end{align*}

For the first system, we find that $x_3 = [1 - 1]_2 = 0$. Then,
\begin{align}
    x_1 = \frac{1 - 1 - 0}{2} = 0 \\
    x_2 = \frac{1 + 1 - 0}{2} = 1
\end{align}

For the second system, we similarly find that $x_6 = [2 - 1]_2 = 1$. Then,
\begin{align}
    x_4 = \frac{2 - 1 - 1}{2} = 0 \\
    x_5 = \frac{2 + 1 - 1}{2} = 1
\end{align}

Thus, we have $x = \begin{bmatrix} 0 & 1 & 0 & 0 & 1 & 1 & 1 & 0\end{bmatrix}^T$.

\textbf{A Modular Arithmetical Property of the CRT Matrix:} Consider the following equation with powers of two in 8 variables and the corresponding $4\times 8$ CRT matrix, say $A$.

\begin{equation}
    w_b^Tx = \sum_{i=1}^8 2^{i-1}x_i = 0
\end{equation}

\begin{align}
&\begin{bmatrix*}[l]
        [2^0]_{3} & [2^1]_{3} & [2^2]_{3} & [2^4]_{3} & [2^5]_{3} & [2^6]_3 & [2^7]_3 \\
        [2^0]_{5} & [2^1]_{5} & [2^2]_{5} & [2^4]_{5} & [2^5]_{5} & [2^6]_5 & [2^7]_5 \\
        [2^0]_{7} & [2^1]_{7} & [2^2]_{7} & [2^4]_{7} & [2^5]_{7} & [2^6]_7 & [2^7]_7 \\
        [2^0]_{11} & [2^1]_{11} & [2^2]_{11} & [2^4]_{11} & [2^5]_{11} & [2^6]_{11} & [2^7]_{11}
        \end{bmatrix*} \\
        &\hphantom{aaaaaaaaaa}= 
        \begin{bmatrix*}[r]
        1 & 2 & 1 & 2 & 1 & 2 & 1 & 2 \\
        1 & 2 & 4 & 3 & 1 & 2 & 4 & 3 \\
        1 & 2 & 4 & 1 & 2 & 4 & 1 & 2 \\
        1 & 2 & 4 & 8 & 5 & 10 & 9 & 7
        \end{bmatrix*}_{4 \times 8}
\end{align}

For the prime $p_i$, the elements in the row $i$ are congruent to the elements of $w_b$. The $i^\text{th}$ element of the $Ax$ vector should be divisible by $p_i$ whenever $w_b^T x = 0$. For instance, one can pick $x = \begin{bmatrix}
    2 & 1 & 1 & 3 & 0 & 1 & -1 & 0
\end{bmatrix}^T$
as a solution for $w_b^Tx = 0$ and we see that $Ax = \begin{bmatrix}
    12 & 15 & 14 & 33
\end{bmatrix}^T = \begin{bmatrix}
    4\cdot 3 & 3\cdot 5 & 2 \cdot 7 & 3 \cdot 11
\end{bmatrix}^T$. This property is essential in the construction of small weight depth-2 circuits for arbitrary threshold functions while the $\text{RMDS}_q$ matrices do not behave in this manner necessarily.
\newpage

\end{document}